\documentclass[12pt]{elsarticle}

\usepackage{amsmath,amsfonts}
\usepackage{amssymb,amsthm}
\usepackage{makecell}
	
\newtheorem{theorem}{Theorem}[section]
\newtheorem{lemma}{Lemma}[section]
\newtheorem{definition}{Definition}[section]
\newtheorem{corollary}{Corollary}[section]
\newtheorem{proposition}{Proposition}[section]
\newtheorem{example}{Example}[section]
\newtheorem{remark}{Remark}[section]
\newtheorem{conjecture}{Conjecture}[section]

\usepackage[commentmarkup=uwave,defaultcolor=red]{changes}
\definechangesauthor[name={Chang-An Zhao},color=red]{zca} 
\definechangesauthor[name={Chang-An},color=blue]{zca2} 

\begin{document}
	
\begin{frontmatter}
\title{On Iso-Dual MDS Codes From Elliptic Curves}
\author{Yunlong Zhu}
\affiliation{
            organization={Department of Mathematics, School of Mathematics},
            addressline={Sun Yat-sen University}, 
            city={Guangzhou},
            postcode={510275}, 
            state={Guangdong},
            country={China}}
\author{Chang-An Zhao}
\affiliation{
            organization={School of Mathematics, Sun Yat-sen University},
            addressline={Sun Yat-sen University}, 
            city={Guangzhou},
            postcode={510275}, 
            state={Guangdong},
            country={China}}
\affiliation{
            organization={Guangdong Key Laboratory of Information Security},
            addressline={}, 
            city={Guangzhou},
            postcode={510006}, 
            state={Guangdong},
            country={China}}
\date{\today}

\begin{abstract}
For a linear code $C$ over a finite field, if its dual code $C^{\perp}$ is equivalent to itself, then the code $C$ is said to be {\it isometry-dual}. In this paper, we first confirm a conjecture about the isometry-dual MDS elliptic codes proposed by Han and Ren. Subsequently, two constructions of isometry-dual maximum distance separable (MDS) codes from elliptic curves are presented. The new code length $n$ satisfies $n\le\frac{q+\lfloor2\sqrt{q}\rfloor-1}{2}$ when $q$ is even and $n\le\frac{q+\lfloor2\sqrt{q}\rfloor-3}{2}$ when $q$ is odd. Additionally, we consider the hull dimension of both constructions. In the case of finite fields with even characteristics, an isometry-dual MDS code is equivalent to a self-dual MDS code and a linear complementary dual MDS code. Finally, we apply our results to entanglement-assisted quantum error correcting codes (EAQECCs) and obtain two new families of MDS EAQECCs.
\end{abstract}

\end{frontmatter}

\section{Introduction}
Let $C$ be an $[n,k]$ linear code defined over $\mathbb{F}_q$. Then it can be seen as a subspace of $\mathbb{F}_q^n$ with dimension $k$. The minimum distance $d$ of $C$ satisfies the Singleton bound, that is
\[
	d\le n-k+1.
\]
When the equality holds, we say that $C$ is an MDS code. The study of MDS codes is one of the central topics in coding theory. For more details see \cite{ttecc,xingcoding}. It is well-known that algebraic geometry codes (AG codes for short) constitute a famous family of linear codes introduced by Goppa in \cite{Goppa}, and the construction of non-Reed-Solomon MDS codes from algebraic curves has attracted significant attention \cite{chenMunuera, Chen-mds-elliptic,Han-tight,Han-maximal,munuera1992main,Munuera1993OnME}. Among AG MDS codes, those derived from elliptic curves have attracted considerable attention, as the rational places of an elliptic function field naturally form an abelian group. Munuera \cite{munuera1992main} showed that the length $n$ of an $[n,k]$ MDS elliptic code should satisfy
\[
	n\le \frac{q+1}{2}+\sqrt{q}+k.
\]
Then Li {\it et al.} \cite{Li} showed that
\[
	n\le (\frac{2}{3}+\epsilon)q
\]
for $\epsilon>0$, $q>\frac{4}{\epsilon^2}$ and $k>C_{\epsilon}ln(q)$, where $C_{\epsilon}$ is a positive constant. Moreover, Li {\it et al.} \cite{Li} proposed a conjecture that for any $\epsilon>0$, there exists a $C_{\epsilon}$ such that for any $q>C_{\epsilon}$, the length $n$ satisfying
\[
	n\le (\frac{1}{2}+\epsilon)q.
\]
More recently, Han and Ren \cite{Han-tight} have proved that 
\[
	n\le \frac{q+1}{2}+\sqrt{q}
\]
when $3\le k\le\frac{q+1-2\sqrt{q}}{10}$. Han and Ren also proposed a conjecture that the upper bound in the above is tight for sufficiently large $q$ and the dimension $k$ with $3\le k\le n-3$.

For two linear codes $C_1$ and $C_2$, we say that they are equivalent if there exists a vector ${\bf v}\in(\mathbb{F}^*_q)^n$ such that
\[
	C_2=\{(c_1,\ldots,c_n)|c_i = v_ix_i, {\bf x}=(x_1,\ldots,x_n)\in Per(C_1)\}
\]
where $Per(C_1)$ is the code obtained from $C_1$ by a coordinate permutation. If $C$ is equivalent to $C^{\perp}$, then $C$ is called {\it isometry-dual} (iso-dual for short). For AG codes, coordinate permutation corresponds to reordering evaluation points, thus we will consider the equivalence of linear codes in this paper without permutation. The iso-dual AG codes were first studied in full generality in \cite{Stich-iso}, and reference \cite{Stich-iso1} proves that the class of iso-dual codes attains the Tsfasman-Vladut and Zink bound over a finite field of quadratic cardinality. Kim and Lee \cite{Kim} presented a construction of iso-dual codes of length $2n+2$ from iso-dual codes of length $2n$. Very recently, Chara {\it et al.} \cite{Chara} gave constructions of iso-dual AG codes from lifting iso-dual codes in rational function fields to elementary abelian $p$-extensions. Bras {\it et al.} also considered the iso-dual in flags of AG codes \cite{Brasdcc,Brastit,Brassiam}.

The main motivation to consider iso-dual MDS codes is the hull-variation problem from MDS linear codes \cite{Chen-hullVaria}. The $hull$ of $C$ is the intersection of $C$ and $C^{\perp}$, which is denoted by $Hull(C)$. If $Hull(C)=C$, then we call $C$ self-dual or self-orthogonal. It is clear that ``self-dual" is a particular case of ``iso-dual", as also in \cite{Sok-iso}. In contrast, if $Hull(C)=\{0\}$, then $C$ is said to be LCD. The hull of linear codes has been widely studied in the literature since it plays a key role in determining whether two given linear codes are equivalent or not \cite{Leon,Leon-per,Sendrier}. It is also interesting in the construction of entanglement-assisted quantum error-correcting codes (EAQECCs, introduced in \cite{Brun}). For further reading about EAQECCs, see the following references \cite{ChenB,Chen-hullVaria,Gao,Jin-self-or,Pereira,Sok-mds,Sok-hull}. 

{\it Our Contributions:} In this paper, we will mainly focus on iso-dual MDS codes from elliptic curves. We first prove that Conjecture 3.7 in \cite{Han-maximal} holds for iso-dual MDS codes. Then we give two constructions of iso-dual MDS elliptic codes over finite fields with even and odd characteristics, respectively. As demonstrated in \cite{Jin-self-dual,Sok-self-dual}, self-dual codes can be derived from iso-dual codes over finite fields with even characteristics. Furthermore, the techniques in \cite{Chen-hull,Mesnager} enable the construction of LCD codes from these self-dual codes. It means that in the case of $\mathbb{F}_q$ with $q$ even, an iso-dual MDS code is both equivalent to a self-dual MDS codes and an LCD MDS codes. We will show more details in Remark IV.1 of Section IV. In contrast, for finite fields with odd characteristics, it is still hard to construct long self-dual MDS codes, and we will give some examples for considerations.

The paper is organized as follows. In Section II, we recall some definitions and results for our purpose, including AG codes, elliptic curves and MDS elliptic codes. Then in Section III, we prove the maximal length of an iso-dual MDS code is less than $\frac{q+1}{2}+\sqrt{q}$. In the following two sections, we give two constructions of iso-dual MDS codes over $\mathbb{F}_q$ with $q$ even in Section IV, and $q$ odd in Section V. At the end of this paper, we apply these iso-dual codes to EAQECCs and obtain two families of MDS EAQECCs.

\section{Preliminaries}
In this section, we recall some basic definitions and useful results for our purpose, which mainly refer to \cite{Silverman,Stich}. The EAQECCs part will be introduced in Section VI.
\subsection{Algebraic geometry code}
Suppose that $X$ is an absolutely irreducible, smooth algebraic curve defined over $\mathbb{F}_q$ of genus $g(X)$. Denote by $\mathbb{F}_q(X)$ the function field generated by $X$. A divisor is defined by the formal sum of places
\[
	G=\sum\limits_Pn_{P}P
\]
where $n_{P}\in\mathbb{Z}$ for the place $P$. If $n_{P}\ge 0$ for all places, then $G$ is called an effective divisor. Let $P_1,\ldots,P_n$ be $n$ pairwise distinct rational places of $\mathbb{F}_q(X)$ and put $D:=P_1+\cdots+P_n$. We also suppose that $G$ is a divisor of $\mathbb{F}_q(X)$ where $P_i\notin \text{\rm supp}(G)$ for all $i$. Then the AG code $C_L(D,G)$ associated with $D$ and $G$ is defined as:
\[
	C_L(D,G) := \{(f(P_1),\ldots,f(P_n)) : f\in\mathcal{L}(G)\}
\]
where
\[
	\mathcal{L}(G) := \{f\in \mathbb{F}_q(X)\backslash\{0\} : \text{\rm div}(f)+G\ge 0\} \cup \{0\}.
\]
It is clear that $\mathcal{L}(G)$ is a linear space over $\mathbb{F}_q$ with dimension $\ell(G)$. Therefore $C_L(D,G)$ is an $[n,k,d]$ linear subspace of $\mathbb{F}_q^n$ with parameters
\[
	k=\ell(G)-\ell(G-D), d\ge n-\text{\rm deg}(G),
\]
and $d^*=n-\text{\rm deg}(G)$ is called the {\it design distance} of $C_L(D,G)$. By the Riemann-Roch theorem, if $2g(X)-1\le \text{\rm deg}(G)<n$, then we can obtain
\[
	k= \text{\rm deg}(G)-g(X)+1.
\]

Another code associated with the divisors $D$ and $G$ is $C_{\Omega}(D,G)$. Let $\Omega$ be the module of rational differentials of $\mathcal{X}$, which is a vector space of dimension one over $\mathbb{F}_q$. For any $\omega\in\Omega$, let $(\omega)$ be the divisor determined by $\omega$. Suppose that $A$ is a divisor of $F(\mathcal{X})$. Then we can define the Weil differential space
\[
    \Omega(A):=\{\omega : \omega=0\ or\ (\omega)\ge A\}.
\]
Then we can define
\[
    C_{\Omega}(D,G):=\{(\omega_{P_1}(1),\ldots,\omega_{P_n}(1)) : \omega\in\Omega(G-D)\}.
\]
The code $C_{\Omega}(D,G)$ is also referred to as an $[n,k',d']$ AG code with parameters
\[
	k'= i(G-D)-i(G), d'\ge \text{\rm deg}(G)-2g(X)+2.
\]
Also by the Riemann-Roch theorem, if $2g(X)-1\le \text{\rm deg}(G)<n$, then we have
\[
	k'= n+g(X)-1-\text{\rm deg}(G).
\]

The relationship between $C_L(D,G)$ and $C_{\Omega}(D,G)$ is elaborated in the following proposition, which generalizes Proposition 2.2.10 and Proposition 8.1.2 in \cite{Stich}:
\begin{proposition}\label{dualAG}
Let $\eta$ be a Weil differential such that $v_{P_i}(\eta)=-1$ and $\eta_{P_i}(1)=1$ for $i=1,\ldots,n$. Then
\[
	C_L(D,G)^{\perp}=C_{\Omega}(D,G)=C_L(D,D-G+(\eta)).
\]
Moreover, if there exists an element $h\in\mathbb{F}_q(X)$ such that $v_{P_i}(h)=1$ for $i=1,\ldots,n$, then
\[
	C_L(D,G)^{\perp}=C_L(D,D-G+(dh)-(h))
\]
where $dh$ is the derivation of $h$.
\end{proposition}
From Proposition 2.6 in \cite{Chara}, if $(\eta)$ is equivalent to $D-2G$, then $C_L(D,G)$ is an iso-dual code. In this paper, we will use the notation for the generalized AG code to determine the explicit vector.

For a given vector ${\bf v}=(v_1,\ldots,v_n)\in(\mathbb{F}_q^*)^n$, the {\it generalized AG} code with $D,G$ and ${\bf v}$ is defined by
\[
	C_L(D,G,{\bf v}) := \{(v_1f(P_1),\ldots,v_nf(P_n)) : f\in\mathcal{L}(G)\}
\]  
and
\[
	C_{\Omega}(D,G,{\bf v}):=\{(v_1\omega_{P_1}(1),\ldots,v_n\omega_{P_n}(1)) : \omega\in\Omega(G-D)\}.
\]
It immediately follows from the above definitions and Proposition \ref{dualAG} that
\[
	C_L(D,G,{\bf v})^{\perp}=C_{\Omega}(D,G,{\bf v^{-1}})
\]
where ${\bf v^{-1}}=(v_1^{-1},\ldots,v_n^{-1})$. Thus $C_L(D,G)$ is iso-dual with vector ${\bf v}$ means
\[
	C_L(D,G)^{\perp}=C_L(D,G,{\bf v}).
\]

\subsection{Elliptic curves}
Let $\mathcal{E}$ be an elliptic curve over $\mathbb{F}_q$ with Weierstrass equation:
\[
	\mathcal{E}: y^2+a_1xy+a_3y=x^3+a_2x^2+a_4x+a_6,
\]
and $a_i\in\mathbb{F}_q$. It is well known that the Jacobian of $\mathcal{E}$ is isomorphic to itself. Therefore the rational places in $\mathbb{F}_q(\mathcal{E})$ form an abelian group, and each rational place corresponds to a rational point. We denote $\mathcal{E}(\mathbb{F}_q)$ as the set of points on $\mathcal{E}$, and $\mathcal{O}$ as the point corresponding the place at infinity. Then $\mathcal{O}$ is also the zero element in the abelian group formed by $\mathcal{E}(\mathbb{F}_q)$. In this paper, we denote by $\oplus$ the addition in the group $\mathcal{E}(\mathbb{F}_q)$, and by $+$ the addition of places. We will not distinguish between a point $P$ and a rational place $P$. Moreover, we denote by $[n]P,nP$ the scalar multiplication of elements in the group and function field, respectively. 

The group structure for elliptic curves is given by the following two lemmas in \cite{Ruck}.
\begin{lemma}\cite[Theorem 1a]{Ruck}\label{T.2.1}
All the possible orders $|{\bf E}({\bf F}_q)|$ of an elliptic curve ${\bf E}$ defined over ${\bf F}_q$, where $q=p^n$ is a prime power,   are given by $$ |{\bf E}({\bf F}_q)|=1+q-\beta,$$
where $\beta$  is an integer with  $|\beta| \leq 2\sqrt{q}$  satisfying one of the following conditions:\\
(a) $\gcd(\beta,p)=1$;\\
(b) If $n$ is even: $\beta=\pm2\sqrt{q}$;\\
(c) If $n$ is even and $p \neq 1$ $mod$ $3$: $\beta=\pm\sqrt{q}$;\\
(d) If $n$ is odd and $p=2$ or $3$: $\beta=\pm p^{\frac{n+1}{2}}$;\\
(e) If either (i) $n$ is odd or (ii) $n$ is even, and $p\neq 1$ $mod$ $4$: $\beta=0$.
\end{lemma}
\begin{lemma}\cite[Theorem 3]{Ruck}\label{T.2.2}
Let ${\bf E}$ be an elliptic curve over a finite field ${\bf F}_q$ with $q=p^n$ elements. Let $|{\bf E}({\bf F}_q)|= \prod_{l} l^{h_l}$ be the prime factorization. Then all the possible groups ${\bf E}({\bf F}_q)$ with the order $|{\bf E}({\bf F}_q)|$ are the following, 
\[
{\bf Z}/p^{h_p}{\bf Z} \times \prod_{l\neq p} ({\bf Z}/l^{a_l}{\bf Z}\times {\bf Z}/l^{h_l-a_l}{\bf Z}),
\]
with\\
(a) In case (b) of Theorem \ref{T.2.1}: Each $a_l$  is equal to $\frac{h_l}{2}$;\\
(b) In cases (a), (c), (d), (e) of Theorem \ref{T.2.1}: $a_l$  is an arbitrary integer satisfying
$0 \leq a_l \leq \min \{v_l(q-1), [\frac{h_l}{2}]\}$, where $v_l(q-1)$ is the order of prime factor $l$ in $q-1$.
\end{lemma}

Considering a divisor $D=\sum_{i=1}^n a_iP_i$ where $P_i$ are rational places, and define
\begin{align*}
	{\rm deg}(D)&=\sum_{i=1}^n a_i,\\
	{\rm sum}(D)&=[a_1]P_1\oplus\cdots\oplus[a_n]P_n.
\end{align*}
Then the important Abel-Jacobi theorem for elliptic curves is given as follows, which serves as the foundation of this paper.
\begin{lemma}\cite[Theorem 11.2]{Washington}\label{abel-jacobi}
Let $D=\sum_{i=1}^n a_iP_i$ where $P_i$ are all rational places be a divisor with ${\rm deg}(D)=0$. Then $D$ is a principal divisor, i.e., there exists a function $f$ such that
\[
	(f)=D
\]
if and only if ${\rm sum}(D)=\mathcal{O}$.
\end{lemma}
Let us present a useful lemma in the following.
\begin{lemma}\label{dx}
The divisor of the differential $dx$ of $\overline{\mathbb{F}}_q(\mathcal{E})$ is
\begin{itemize}
\item if $q$ odd, then $(dx)=(y)$,\\
\item if $q$ even and $\mathcal{E}$ has the form $y^2+a_3y=x^3+a_4x+a_6$, then $(dx)=0$,\\
\item if $q$ even and $\mathcal{E}$ has the form $y^2+xy=x^3+a_2x^2+a_6$, then $(dx)=(x)$.
\end{itemize}
\end{lemma}
\begin{proof}
These assertions follow easily from Theorem 4.3.2 and Proposition 6.1.3 in \cite{Stich}.
\end{proof}
Note that $(dx)=0$ means there is no place ramified on $\mathbb{F}_q(\mathcal{E})$ without $\mathcal{O}$. Therefore the group $\mathcal{E}(\mathbb{F}_q)$ has no element with order 2, i.e., $\#\mathcal{E}(\mathbb{F}_q)$ must be odd. In this paper, we mainly consider the case of $(dx)=(x)$ over finite fields with even characteristics.

\subsection{MDS codes from elliptic curves}
Suppose that $P_1,\ldots,P_n$ are $n$ different points of $\mathcal{E}(\mathbb{F}_q)$ and $G$ is an effective divisor such that $P_i\notin \text{\rm supp}(G)$ for all $i$. Let $D=P_1+\cdots+P_n$ be the divisor defined by the $n$ points, then we have $C_L(D,G)$ is an $[n,\text{\rm deg}(G),\ge n-\text{\rm deg}(G)]$ code since $g(\mathcal{E})=1$. It is well known that $C_L(D,G)$ and $C_{\Omega}(D,G)$ are both NMDS or MDS. Then the following lemma gives the condition to check whether $C_L(D,G)$ is an MDS code.
\begin{lemma}\label{mds-condition}\cite{Chen-mds-elliptic}
Suppose that $G$ can be represented by $G=Q_1+\cdots+Q_k$ where $Q_1,\ldots,Q_k$ are $k$ points and distinct from $P_i$ for all $i$. Then $C_L(D,G)$ is MDS if and only if
\[
	P_{i_1}\oplus\cdots\oplus P_{i_k}\neq Q_1\oplus\cdots\oplus Q_k
\]
for any $k$ points $P_{i_1},\ldots,P_{i_k}$ in $\text{\rm supp}(D)$.
\end{lemma}

Now we present a construction of MDS elliptic codes in the following lemma, which also can be seen as a generalization of the construction given in \cite{Chen-mds-elliptic}.
\begin{lemma}\label{mds_const}
Suppose that $Q_1,\ldots,Q_s$ are $s$ distinct points such that the order of $Q_j$ is prime to the order of $P_i$ for all $1\le i\le n$ and $1\le j\le s$. Denote by
\[
	\mathcal{Q}_j=\{Q_j\oplus P_i:1\le i\le n\}.
\]
Let $D=\sum\limits_{j=1}^n\sum\limits_{P\in\mathcal{Q}_j}P$, and $G=k\mathcal{O}$. If 
\[
	[k_1]Q_1\oplus\cdots\oplus [k_s]Q_s\neq\mathcal{O}
\]
for all
\[
	k_1+\cdots+k_s=k.
\]
Then the code $C_L(D,G)$ is an $[ns,k]$ MDS code.
\end{lemma}
\begin{proof}
For any $k$ points in $\text{\rm supp}(D)$, they can be partitioned into $s$ parts as subsets of $\mathcal{Q}_j$. Therefore the sum of these points in group has the form as
\[
	[k_1]Q_1\oplus\cdots\oplus [k_s]Q_s\oplus P'
\]
where $k_1+\cdots+k_s=k$. Since the order of $Q_j$ is prime to $P_i$, we must have the order of $Q_j$ is also prime to $P'$. If
\[
	[k_1]Q_1\oplus\cdots\oplus [k_s]Q_s\neq\mathcal{O},
\]
then the sum of these $k$ points must not be equal to $\mathcal{O}$. By Lemma \ref{mds-condition}, the code $C_L(D,G)$ is MDS.
\end{proof}
The following result is an immediate corollary of the lemma in the above.
\begin{corollary}\label{mdsmul}
With the notations in Lemma \ref{mds_const}. Suppose that there exists $t$ points $Q'_{1},\ldots,Q'_{t}$ such that $Q'_j\notin\text{\rm supp}(D)$ and the order $q_j$ of $Q'_j$ is prime to $P_i$ for all $1\le i\le n$. Let
\[
	\tilde{G}=m_{1}Q'_{1}+\cdots+m_{t}Q'_{t}+(k-m_{1}-\cdots-m_{t})\mathcal{O}.
\]
If 
\[
	[k_1]Q_1\oplus\cdots\oplus [k_s]Q_s\oplus [q_{1}-m_{1}]Q'_1\oplus\cdots\oplus [q_{t}-m_{t}]Q'_t\neq\mathcal{O}
\]
for all
\[
	k_1+\cdots+k_{s}=k,
\]
then the code $C_L(D,\tilde{G})$ is an $[ns,k]$ MDS code.
\end{corollary}

A conjecture regarding the maximal length of MDS codes from elliptic curves given by Han and Ren in \cite{Han-tight} is cited as follows:
\begin{conjecture}
Let $C_L(D,G)$ be an $[n,k]$ MDS code from an elliptic curve $\mathcal{E}$. If $q$ is sufficiently large and $3\le k\le n-3$, then $n\le\frac{\#\mathcal{E}(\mathbb{F}_q)}{2}$. 
\end{conjecture}
Han and Ren also proved that this upper bound is tight for $q\ge 289$ and $3\le k\le \frac{q+1-2\sqrt{q}}{10}$ with $q$ odd and $k$ odd in \cite{Han-maximal}. Suppose that $G$ is a finite abelian group and $S$ is a non-empty subset of elements of $G$. Denote by
\[
	\Sigma_k(S) = \{g\in G:g=a_{i_1}+\cdots+a_{i_k}\ {\rm where}\ a_{i_j}\in S\},
\]
and $G[2]$ the 2-torsion subgroup of $G$. For our purpose, we present a useful lemma given in \cite{Han-maximal} as follows.
\begin{lemma}\label{sigma3}
Let $A$ be a subset of an abelian group $G$ with 
\[
	\#A >{\rm max}\{\frac{2}{5}\#G, 12\#G[2] +54\}.
\]
Then either $\Sigma_3(A)=G$, or $A$ is contained in a coset of an index-two subgroup of $G$.
\end{lemma}

\section{Maximal length of iso-dual MDS elliptic codes}
Before giving the main results, we first present the following important lemma for our purpose.
\begin{lemma}\label{isodual-max}
Let $G$ be a finite abelian group with an even order, and let $g$ be an element of $G$. Suppose that $n\ge\frac{\#G}{2}+1$ is an even number, and $A=\{a_1,\ldots,a_n\}$ is a subset with
\begin{itemize}
\item $g\notin A$,
\item $a_1+\cdots+a_n=2g$.
\end{itemize}
Then we have $g\in\Sigma_{k}(A)$ with $k=\frac{n}{2}$.
\end{lemma}
\begin{proof}
Since the order of $G$ is even, then $G$ has an index two subgroup $G_1$ and there exists an element $g_2\in G$ such that $G=G_1\cup (g_2+G_1)$. Denote by $G_2=g_2+G_1$ and $n_i=\#A\cap G_i$. Then we have $n_i>0$ since $\#G_i=\frac{\#G}{2}$. Moreover, at least one of $n_1$ and $n_2$ is not less than $k$. W.o.l.g, suppose that $n_1\ge k$. Note that $2g\in G_1$. It follows that
\[
	b_1+\cdots+b_{n_2}\in G_1
\]
where $b_j\in G_2$ for $1\le j\le n_2$. We consider that $n_2\le k-3$ and distinguish two cases, namely $g\in G_1$ and $g\in G_2$.
\begin{itemize}
\item Case I: $g\in G_1$. Since $n_2\le k-3$, there exists a subset $B\subset A$ such that $\#B\cap G_2=n_2$ and $\#B\cap G_1=k-3-n_2$. Then $\#A\backslash B=k+3$ and $A\backslash B\subset G_1$. We can obtain
\[
	k+3>\frac{\#G+10}{4}\ge\frac{\#G}{5}=\frac{2}{5}\#G_1.
\]
Assume that the sum of all elements in $B$ is $h$. Thus from Lemma \ref{sigma3}, we can find three elements $c_1,c_2,c_3$ in $A\backslash B$ such that $c_1+c_2+c_3=g-h$. Therefore $g\in\Sigma_{k}(A)$.
\item Case II: $g\in G_2$. Similarly, there exists a subset $B\subset A$ such that $\#B\cap G_2=n_2-1$ and $\#B\cap G_1=k-2-n_2$. Then $\#(A\backslash B)\cap G_1=k+2$. We can also obtain that $k+2\ge\frac{2}{5}\#G_1$. Since $g-h\in G_1$, it follows that $g\in\Sigma_{k}(A)$ by Lemma \ref{sigma3}.
\end{itemize}
If $n_2\ge k-2$, then $g\in\Sigma_{k}(A)$ is trivial since any $k+3$ elements of $A$ is not contained in either $G_1$ or $G_2$. Thus the proof is completed.
\end{proof}
It is worth noting that every canonical divisor is principal in elliptic function fields (see Lemma 10 in \cite{Jin-self-dual}). Now, we present our result as follows.
\begin{theorem}
Suppose that $C_L(D,G)$ is a $q$-ary $[n,k]$ iso-dual MDS code from elliptic curve $\mathcal{E}$. Then we have $n\le\frac{\#\mathcal{E}(\mathbb{F}_q)}{2}$.
\end{theorem}
\begin{proof}
Since $C_L(D,G)$ is iso-dual, it follows that $G$ is equivalent to $D-G+(\eta)$, i.e., 
\[
	D-2G\sim -(\eta).
\]
We can obtain that $D-2G$ is a principal divisor. Suppose that ${\rm supp}(D)=\{P_1,\ldots,P_n\}$ and ${\rm sum}(G)=Q$. Then we have
\[
	P_1\oplus\cdots\oplus P_n=[2]Q.
\]
By Lemma \ref{isodual-max}, if $n\ge\frac{\#\mathcal{E}(\mathbb{F}_q)}{2}+1$, then we can find $k$ points in ${\rm supp}(D)$ such that
\[
	P_{i_1}\oplus\cdots\oplus P_{i_k}=Q.
\]
It is contradictory to $C_L(D,G)$ is MDS by Lemma \ref{mds-condition}.
\end{proof}

\section{MDS Elliptic Codes Over Fields of even Characteristics}
In this section, we assume that $q$ is a power of 2. Then we will present a construction of iso-dual MDS codes over $\mathbb{F}_q$. We also consider the hull dimension of these codes.

{\it Construction 1:} Suppose that $\mathcal{E}$ is an elliptic curve defined over $\mathbb{F}_q$ with even order. Then there exists a non-zero point $Q_1=(0,\gamma_1)$ with order 2. Let $k$ be an even number, and $n=2k$. We define a point set $\mathcal{P}=\{P_{\alpha_i}^+,P_{\alpha_i}^-| i=1,\ldots,k\}$ such that each point has odd order in $\mathcal{E}(\mathbb{F}_q)$ with
\[
P_{\alpha_i}^+\oplus P_{\alpha_i}^-=\mathcal{O}
\]
for $i=1,\ldots,k$. Denote by
\[
	\{P_1,\ldots,P_n\} := \{Q_1\oplus P_{\alpha_i}^+,Q_1\oplus P_{\alpha_i}^-| i=1,\ldots,k\}.
\]
Let $D=P_1+\cdots+P_n$ and $G=(k-1)\mathcal{O}+Q_1$ be two effective divisors, and let $C_L(D,G)$ be the AG code associated with $D$ and $G$. Then we have the following results.

\begin{theorem}\label{isomdseven}
The code $C_L(D,G)$ given in Construction 1 is an $[n,k,n-k+1]$ iso-dual code.
\end{theorem}
\begin{proof}
Firstly we need to prove that $C_L(D,G)$ is MDS. Since $k$ is an even number and $Q_1$ have order 2, it is clear that
\[
	[k]Q_1\oplus [2-1]Q_1\neq\mathcal{O}.
\]
From Corollary \ref{mdsmul} we can obtain that $C_L(D,G)$ is MDS. Then we consider that $C_L(D,G)$ is iso-dual. Note that for any $i=1,\ldots,k$, we have
\[
	Q_1\oplus P_{\alpha_i}^+\oplus Q_1\oplus P_{\alpha_i}^-=\mathcal{O}
\]
and
\[
	Q_1\oplus P_{\alpha_i}^+\neq Q_1.
\]
Therefore there exists a non-zero element $\alpha_i^1$ such that
\[
Q_1\oplus P_{\alpha_i}^+=(\alpha_i^1,\beta_i^+),Q_1\oplus P_{\alpha_i}^-=(\alpha_i^1,\beta_i^-).
\]
Let $h=\prod_{i=1}^{k}(x-\alpha_i^1)$, then 
\[
	(h) = P_1+\cdots+P_n-n\mathcal{O}.
\]
Thus from Proposition \ref{dualAG} and Lemma \ref{dx}, we have
\[
	(\eta)=(h'_x)+(dx)+n\mathcal{O}-D,
\]
and 
\begin{align*}
	D-G+(\eta)&=D-(k-1)\mathcal{O}-Q_1+(h'_x)+(x)-D+n\mathcal{O}\\
 	&=(k+1)\mathcal{O}-Q_1+(h'_x)+2Q_1-2Q_1\\
	&=(k-1)\mathcal{O}+Q_1+(h'_x).
\end{align*}
It follows that
\begin{align*}
	C_L(D,D-G+(\eta))&=C_L(D,G+(h'_x))\\
	&=C_L(D,G,{\bf v})
\end{align*}
where ${\bf v}=\{v_1,\ldots,v_n\}$ and $v_i=\frac{1}{h'_x}(P_i)$. The proof is completed.
\end{proof}
We will present some examples of Construction 1 following the result below.
\begin{corollary}\label{hulleven}
Suppose that $C_L(D,G)$ is given in Construction 1 with parameters $[n,k,n-k+1]$ and $n\ge8$. Then we have 
\[
{\rm dim}(C_L(D,G)\cap C_L(D,G)^{\perp})\le k-1.
\]
\end{corollary}
\begin{proof}
Suppose that ${\rm dim}(C_L(D,G)\cap C_L(D,G)^{\perp})=k$. It means that for any function $f\in\mathcal{L}(G)$, there exists a function $g\in\mathcal{L}(G)$ such that
\[
	f(P_i)= \frac{1}{h'_x}g(P_i),
\]
i.e.,
\[
	(h'_xf-g)(P_i)=0
\]
for $i=1,\ldots,n$. Assume that the function $y-\gamma_1$ has zeros $Q_1,Q_2,Q_3$. Then we have
\[
	(\frac{y-\gamma_1}{x})=Q_2+Q_3-Q_1-\mathcal{O}.
\]
Therefore $\mathcal{L}(G)$ has a basis as
\[
	\{x^{\frac{k-2}{2}},\ldots,x,1,\frac{y-\gamma_1}{x},\ldots,(y-\gamma_1)x^{\frac{k-4}{2}}\}.
\]
We consider the function $f=1$. From the basis we can obtain that $xg\in\mathcal{L}((k+1)\mathcal{O})$ for any $g\in\mathcal{L}(G)$. Since $k$ is an even number, we have ${\rm deg}_x(h'_x)=k-2$. Therefore $(xh'_x)\in\mathcal{L}((n-2)\mathcal{O})$. From the strict triangle inequality of discrete valuation, if $v_{\mathcal{O}}(xh'_x-xg)\neq v_{\mathcal{O}}(xg)$, then
\[
	v_{\mathcal{O}}(xh'_x)=min\{v_{\mathcal{O}}(xh'_x-xg),v_{\mathcal{O}}(xg)\}. 
\]
Notice that $n\ge 8$. Then $n-2>k+1$ and we can obtain
\[
	(xh'_x-xg)_{\infty}\in\mathcal{L}((n-2)\mathcal{O}),x(h'_x-g)\neq0.
\]
It means that $xh'_x-xg$ has at most $n-2$ zeros, and $h'_x-g$ has at most $n-4$ zeros, which is contradictory to $(h'_x-g)(P_i)=0$ for $i=1,\ldots,n$. The proof is completed.
\end{proof}
\begin{remark}
In \cite{Chen-mds-elliptic} and \cite{Jin-self-dual}, the authors have given some constructions of self-dual codes from elliptic curves. Using their results, we can obtain self-dual MDS code from $C_L(D,G)$ over $\mathbb{F}_q$. Assume that ${\bf u}=\{u_1,\ldots,u_n\}$ where $(u_i)^2=\frac{1}{h'_x}(P_i)$ ($q$ is even), then $v_i(u_i)^{-1}=u_i$ for any $1\le i\le n$ and
\[
	C_L(D,G,{\bf u})^{\perp}=C_{\Omega}(D,G,{\bf u^{-1}})=C_L(D,G,{\bf u}).
\]
Thus the generalized AG code $C_L(D,G,{\bf u})$ is a self-dual MDS code. On the other hand, using the technique given in \cite{Chen-hull} and \cite{Chen-hullVaria}, we can obtain LCD MDS code from $C_L(D,G)$ over $\mathbb{F}_q$. Assume that
\[
	{\rm dim}(C_L(D,G)\cap C_L(D,G)^{\perp})=\ell.
\]
Then we can find a vector ${\bf\hat{u}}=\{\hat{u}_1,\ldots,\hat{u}_{n}\}\in(\mathbb{F}_q^*)^n$ satisfying
\begin{itemize}
	\item There are $\ell$ positions such that $\hat{u}_i^2\neq1$.
	\item Otherwise $\hat{u}_j=1$. 
\end{itemize}
Then we have
\[
{\rm dim}(C_L(D,G,{\bf\hat{u}})\cap C_L(D,G,{\bf\hat{u}})^{\perp})=0.
\]
Thus the generalized AG code $C_L(D,G,{\bf\hat{u}})$ is a LCD MDS code. Moreover, suppose that $t=\frac{k}{2}$ be an odd number, then we have
\[
	C_L(D,G)=C_L(D,t\mathcal{O}+tQ_1,{\bf\hat{v}})
\]
where $\bf\hat{v}=(v_1,\ldots,v_n)$ and $v_i=x^{\frac{1-t}{2}}(P_i)$. By \cite{Mesnager} we have that $C_L(D,t\mathcal{O}+tQ_1)$ is equivalent to an LCD code since $t$ is odd and $Q_1$ has order 2. Thus $C_L(D,G)$ is also equivalent to an LCD code.
\end{remark}

\begin{example}
Let $q=16$, and $\mathbb{F}^*_q=<\theta>$ where $\theta^4+\theta+1=0$. Consider the elliptic curve:
\[
	\mathcal{E} : y^2+xy=x^3+\theta^3x^2+\theta^3+1.
\]
We have $\mathcal{E}(\mathbb{F}_q)=22$, and $Q_1=(0,\theta^3+\theta+1)$. From Construction 1, we take 8 points with order 11 and the $x$-coordinate sets of $\{P_1,\ldots,P_8\}$ is
\[
	\{\theta^2 + 1,1,\theta,\theta^2+\theta+1\}.
\]
Define two divisors $D=P_1+\cdots+P_8$ and $G=3\mathcal{O}+Q_1$. Then the code $C_L(D,G)$ is an $[8,4]$ iso-dual MDS code and
\[
C_L(D,G)^{\perp}={\bf v}C_L(D,G)
\]
where ${\bf v}=(\frac{1}{h'_x}(P_1),\ldots,\frac{1}{h'_x}(P_8))$. Moreover, we have
\[
{\rm dim}(C_L(D,G)\cap C_L(D,G)^{\perp})=0,
\]
which means that $C_L(D,G)$ is an LCD MDS code. Let 
\[
	{\bf u}=(\theta^2,\theta^2,\theta+1,\theta+1,\theta,\theta,\theta^2+1,\theta^2+1).
\]
By using Magma,we have
\[
{\rm dim}(C_L(D,G,{\bf u})\cap C_L(D,G,{\bf u})^{\perp})=4.
\]
It means that $C_L(D,G,{\bf u})$ is self-dual. 
\end{example}

\begin{example}
Let $q=64$, and $\mathbb{F}^*_q=<\theta>$ where $\theta^6+\theta^4+\theta^3+\theta+1=0$. Consider the elliptic curve:
\[
\mathcal{E} : y^2+xy=x^3+\theta^3x^2+\theta^3+1.
\]
We have $\mathcal{E}(\mathbb{F}_q)=78$, and $Q_1=(0,\theta^5+\theta^4+\theta^3+\theta^2+\theta)$. From Construction 1, we take $n=36$ points with order 39 and denote by $\{P_1,\ldots,P_n\}$. Define two divisors $D=P_1+\cdots+P_n$ and $G=17\mathcal{O}+Q_1$. Then the code $C_L(D,G)$ is a $[36,18]$ iso-dual MDS code and
\[
C_L(D,G)^{\perp}={\bf v}C_L(D,G)
\]
where ${\bf v}=(\frac{1}{h'_x}(P_1),\ldots,\frac{1}{h'_x}(P_n))$. Moreover, we have
\[
{\rm dim}(C_L(D,G)\cap C_L(D,G)^{\perp})=2.
\]
Let ${\bf u}=(u_1,\ldots,u_n)$ with $u_i^2=\frac{1}{h'_x}(P_i)$. By using Magma,we have
\[
{\rm dim}(C_L(D,G,{\bf u})\cap C_L(D,G,{\bf u})^{\perp})=18.
\]
It means that $C_L(D,G,{\bf u})$ is self-dual. There are only two positions $i_1,i_2$ such that $v_{i_1}=v_{i_2}=1$. Thus let ${\bf\hat{u}}=(1,\ldots,u_{i_1},u_{i_2},\ldots,1)$ with $u_{i_1}=u_{i_2}=\theta$, we have
\[
{\rm dim}(C_L(D,G,{\bf\hat{u}})\cap C_L(D,G,{\bf\hat{u}})^{\perp})=0.
\]
It means that $C_L(D,G,{\bf\hat{u}})$ is LCD.
\end{example}

\begin{example}
Let $q=256$, and $\mathbb{F}^*_q=<\theta>$ where $\theta^8+\theta^4+\theta^3+\theta^2+1=0$. Consider the elliptic curve:
\[
\mathcal{E} : y^2+xy=x^3+\theta^5x^2+\theta^5+\theta^4+1.
\]
We have $\mathcal{E}(\mathbb{F}_q)=286$, and $Q_1=(0,\theta^7+\theta^5+\theta^3+\theta^2+\theta+1)$. From Construction 1, we take $n=140$ points with order 143 and denote by $\{P_1,\ldots,P_n\}$. Define two divisors $D=P_1+\cdots+P_n$ and $G=69\mathcal{O}+Q_1$. Then the code $C_L(D,G)$ is a $[140,70]$ iso-dual MDS code and
\[
C_L(D,G)^{\perp}={\bf v}C_L(D,G)
\]
where ${\bf v}=(\frac{1}{h'_x}(P_1),\ldots,\frac{1}{h'_x}(P_n))$. Moreover, we have
\[
{\rm dim}(C_L(D,G)\cap C_L(D,G)^{\perp})=2.
\]
Let ${\bf u}=(u_1,\ldots,u_n)$ with $u_i^2=\frac{1}{h'_x}(P_i)$. By using Magma,we have
\[
{\rm dim}(C_L(D,G,{\bf u})\cap C_L(D,G,{\bf u})^{\perp})=70.
\]
It means that $C_L(D,G,{\bf u})$ is self-dual. Let ${\bf\hat{u}}=(\hat{u}_1,\ldots,\hat{u}_n)$ such any two values are $\theta$ and others are all 1. Then we have
\[
{\rm dim}(C_L(D,G,{\bf\hat{u}})\cap C_L(D,G,{\bf\hat{u}})^{\perp})=0.
\]
It means that $C_L(D,G,{\bf\hat{u}})$ is LCD.
\end{example}
\begin{remark}
The length $n$ of codes given by Construction 1 is restricted by the number of points with odd order in $\mathcal{E}(\mathbb{F}_q)$. Thus it is restricted by the odd divisor of $\#\mathcal{E}(\mathbb{F}_q)$ by Lemma \ref{T.2.2}. Note that we also need that $k$ is even. We can easily check that if $q\ge16$ is a square, then
\[
	n\le \frac{q+2\sqrt{q}-8}{2}.
\]
Suppose that $q$ is not a square, since $n+1$ must be less than the greatest odd divisor of $\#\mathcal{E}(\mathbb{F}_q)$. When $q+1+\lfloor2\sqrt{q}\rfloor$ is an even number such that
\[
	\frac{q+1+\lfloor2\sqrt{q}\rfloor}{2}-1\equiv0\mod4.
\]
The length $n$ can be chosen to be $\frac{q+\lfloor2\sqrt{q}\rfloor-1}{2}$, and the smallest $q$ in this case is $q=2^{25}$. Otherwise we have
\[
	n\le \frac{q+\lfloor2\sqrt{q}\rfloor-3}{2}.
\]The following example is given for illustrations to show the latter case.
\end{remark}
\begin{example}
Let $q=32$, and $\mathbb{F}^*_q=<\theta>$ where $\theta^5+\theta+1=0$. Consider the elliptic curve:
\[
\mathcal{E} : y^2+xy=x^3+x^2+\theta^2+\theta.
\]
We have $\mathcal{E}(\mathbb{F}_q)=42$. From Construction 1 we can obtain a $[20,10]$ iso-dual MDS code and a $[20,10]$ self-dual MDS code. In this case, we have the code length $n=\frac{q+\lfloor2\sqrt{q}\rfloor-3}{2}$.
\end{example}

\section{MDS Elliptic Codes over Fields of Odd Characteristics}
In this section, we assume that $q$ is always odd. Then we will present a construction of iso-dual MDS codes over $\mathbb{F}_q$. We also consider the hull dimension of these codes.

{\it Construction 2:} Suppose that $\mathcal{E}$ is an elliptic curve defined over $\mathbb{F}_q$ with $\mathcal{E}[2]\subset\mathcal{E}(\mathbb{F}_q)$. Denote by 
\[
	\mathcal{E}[2] := \{\mathcal{O},Q_1=(\beta_1,0),Q_2=(\beta_2,0),Q_3=(\beta_3,0)\}.
\]
Let $k$ be an even number, and $\mathcal{P}=\{P_{\alpha_i}^+,P_{\alpha_i}^-| i=1,\ldots,\frac{k}{2}\}$ be $k$ points such that each point has odd order in $\mathcal{E}(\mathbb{F}_q)$ and 
\[
P_{\alpha_i}^+\oplus P_{\alpha_i}^-=\mathcal{O}
\]for $i=1,\ldots,k$. Denote by
\[
	\{P_1,\ldots,P_n\} := (Q_1\oplus\mathcal{P}) \cup (Q_2\oplus\mathcal{P})
\]
where 
\[
	Q_j\oplus\mathcal{P} := \{Q_j\oplus P_{\alpha_i}^+,Q_j\oplus P_{\alpha_i}^-| i=1,\ldots,\frac{k}{2}\}.
\]
Let $D=P_1+\cdots+P_n$ and $G=(k-1)\mathcal{O}+Q_1$ be two divisors, and let $C_L(D,G)$ be the AG code associated with $D$ and $G$. Then we have the following results.
\begin{theorem}\label{isomdsodd}
The code $C_L(D,G)$ given in Construction 2 is an $[n,k,n-k+1]$ iso-dual code.
\end{theorem}
\begin{proof}
Firstly, we prove that $C_L(D,G)$ is an MDS code. For any partition of $k=k_1+k_2$, since $k$ is an even number, the two integers $k_1$ and $k_2$ are either both even or both odd. It follows that
\[
	[k_1]Q_1\oplus[k_2]Q_2=\mathcal{O}\ or\ Q_3,
\]
which means
\[
	[k_1]Q_1\oplus[k_2]Q_2\oplus[2-1]Q_1\neq \mathcal{O}.
\]
From Corollary \ref{mdsmul} we can obtain that $C_L(D,G)$ is MDS. 

It remains to prove that $C_L(D,G)$ is an iso-dual code. For any $i=1,\ldots,\frac{k}{2}$ and $j=1,2$,
\[
	(Q_j\oplus P_{\alpha_i}^+)\oplus(Q_j\oplus P_{\alpha_i}^-)=\mathcal{O}.
\]
Therefore if $Q_j\oplus P_{\alpha_i}^+=(\alpha_i^j,\beta_i^j)$, then $(Q_j\oplus P_{\alpha_i}^-)=(\alpha_i^j,-\beta_i^j)$. Let
\[
	h=\prod_{i=1}^{\frac{k}{2}}(x-\alpha_i^1)(x-\alpha_i^2).
\]
We have
\[
	(h)=P_1+\cdots+P_n-n\mathcal{O}.
\]
Thus $C_L(D,G)^{\perp}=C_L(D,D-G+(\eta))$ with $(\eta)=(dh)-(h)$ by \cite[Proposition 8.1.2]{Stich}. Since $q$ is odd, then 
\[
	(dx) = -2(x)_{\infty}+\mathcal{O}+Q_1+Q_2+Q_3=(y),
\]
and
\begin{align*}
	D-G+(\eta)&=D-(k-1)\mathcal{O}-Q_1+(h'_x)+(y)-D+n\mathcal{O}\\
 	&=(k+1)\mathcal{O}-Q_1+(h'_x)+(y)+2Q_1-2Q_1\\
	&=(k-1)\mathcal{O}+Q_1-(x-\beta_1)+(h'_x)+(y).
\end{align*}
It follows that
\begin{align*}
	C_L(D,D-G+(\eta))&=C_L(D,G-(x-\beta_1)+(h'_x)+(y))\\
	&=C_L(D,G,{\bf v})
\end{align*}
where ${\bf v}=\{v_1,\ldots,v_n\}$ and $v_i=\frac{x-\beta_1}{h'_xy}(P_i)$. The proof is completed.
\end{proof}
\begin{remark}
Suppose that $\#\mathcal{E}(\mathbb{F}_q)=4r$ where $r$ is odd. Then the length $n$ of any code given in Construction 2 satisfies $n\le 2r-2$. Moreover, if $\mathcal{E}$ is maximal over $\mathbb{F}_q$ and $q$ is square, then we have $n\le\frac{q+2\sqrt{q}-3}{2}$. If $q$ is not square, since $\#\mathcal{E}[2]=4$ and $q$ odd, then we also have $n\le\frac{q+\lfloor2\sqrt{q}\rfloor-3}{2}$.
\end{remark}
We will present some examples of Construction 2 following the result below.
\begin{corollary}
Suppose that $C_L(D,G)$ is given in Construction 2 with parameters $[n,k,n-k+1]$. Then we have ${\rm dim}(C_L(D,G)\cap C_L(D,G)^{\perp})\le k-1$. 
\end{corollary}
\begin{proof}
With the same notations in Construction 2 and Theorem \ref{isomdsodd}, assume that
\[
	{\rm dim}(C_L(D,G)\cap C_L(D,G)^{\perp})= k,
\]
then $C_L(D,G)$ is self-dual. Therefore we can obtain that for any $f\in\mathcal{L}(G)$, there exists a function $g\in\mathcal{L}(G)$ such that
\[
	f(P_i)=\frac{x-\beta_1}{h'y}g(P_i),
\]
i.e.,
\[
	(h'yf-(x-\beta_1)g)(P_i)=0
\]
for $i=1,\ldots,n$. Note that $f(P_i)=0$ if and only if $g(P_i)=0$. We consider the minimum Hamming weight code. By \cite[Corollary 7.4.2]{Huffman-fecc} we have
\[
	A_{n-k+1}=(q-1)\binom{n}{k-1}.
\]
Thus we can find a function $f$ with $k-1$ zeros $P_{i_1},\ldots,P_{i_{k-1}}$ among $\{P_1,\ldots,P_n\}$ such that
\[
	P_{i_1}\oplus\cdots\oplus P_{i_{k-1}}\neq\mathcal{O}
\]
and $(f)_{\infty}=(k-1)\mathcal{O}+Q_1$. From Construction 2, we also have
\[
	P_{i_1}\oplus\cdots\oplus P_{i_{k-1}}\neq Q_1.
\]
Therefore the corresponding function $g$ of $f$ must satisfy $(g)_{\infty}=(k-1)\mathcal{O}+Q_1$. Then we can obtain that
\[
	h'yf-(x-\beta_1)g=f(h'y-a(x-\beta_1))
\]
where $a=\frac{f}{g}\in\mathbb{F}_q$. Since $f$ has $k-1$ zeros among $\{P_1,\ldots,P_n\}$, the function $h'y-a(x-\beta_1)$ must have at least $k+1$ zeros in $\{P_1,\ldots,P_n\}$. Then we can find a pair of points $P_{\alpha}^+=(\alpha,\beta_{\alpha})$ and $P_{\alpha}^-=(\alpha,-\beta_{\alpha})$ such that
\begin{align*}
(h'y-a(x-\beta_1))(P_{\alpha}^+)&=(h'y-a(x-\beta_1))(P_{\alpha}^-)=0,\\
f(P_{\alpha}^+)\neq0&,f(P_{\alpha}^-)\neq0.
\end{align*}
It follows that 
\begin{align*}
	&h'(\alpha)y(\beta_{\alpha})-a(\alpha-\beta_1)\\
	=&-h'(\alpha)y(\beta_{\alpha})-a(\alpha-\beta_1)\\
	=&0,
\end{align*}
which means $h'(\alpha)y(\beta_{\alpha})=0$ and $\alpha-\beta_1=0$. From the construction and the fact that $v_{P_i}(h)=1$ for $i=1,\ldots,n$, we have $h'y(P_i)\neq0$ and $(x-\beta_1)(P_i)\neq0$. Then we get a contradiction and the proof is completed.
\end{proof}
\begin{example}
Let $q=5^2$, and $\mathbb{F}_q^*=<\theta>$ where $\theta^2-\theta+2=0$. Consider the elliptic curve
\[
\mathcal{E}:y^2=x^3+1
\]
defined over $\mathbb{F}_q$. We have $\#\mathcal{E}(\mathbb{F}_q)=36$ and $\mathcal{E}(\mathbb{F}_q)\simeq \mathcal{E}[2]\times \mathcal{E}[3]$. It follows that
\[
	Q_1=(-1,0),Q_2=(2\theta+2,0),Q_3=(3\theta-1,0).
\]
From Construction 2, we take the point set
\[
	\{Q_1\oplus P|P\in \mathcal{E}[3]\backslash\{\mathcal{O}\}\}\cup\{Q_2\oplus P|P\in \mathcal{E}[3]\backslash\{\mathcal{O}\}\}.
\]
Denote the point set by $\{P_1,\ldots,P_n\}$ with $n=16$. Define two divisors $D=P_1+\cdots+P_n$ and $G=7\mathcal{O}+Q_1$. Then the code $C_L(D,G)$ is a $[16,8]$ iso-dual MDS code and
\[
	C_L(D,G)^{\perp}={\bf v}C_L(D,G)
\]
where ${\bf v}=(\frac{x+1}{h'_xy}(P_1),\ldots,\frac{x+1}{h'_xy}(P_n))$. Moreover, we have
\[
	{\rm dim}(C_L(D,G)\cap C_L(D,G)^{\perp})=0,
\]
which means that $C_L(D,G)$ is an LCD MDS code. By using Magma, for any vector ${\bf u}\in(\mathbb{F}_q^*)^n$ we have
\[
	{\rm dim}(C_L(D,G,{\bf u})\cap C_L(D,G,{\bf u})^{\perp})\le2.
\]
\end{example}

\begin{example}
Let $q=7^2$, and $\mathbb{F}_q^*=<\theta>$ where $\theta^2-\theta+3=0$. Consider the elliptic curve 
\[
\mathcal{E}:y^2=x^3+x+3
\]
defined over $\mathbb{F}_q$. We have $\#\mathcal{E}(\mathbb{F}_q)=60$ and $\mathcal{E}[2]\subset\mathcal{E}(\mathbb{F}_q)$. It follows that
\[
	Q_1=(5,0),Q_2=(2\theta,0),Q_3=(5\theta+2,0).
\]
From Construction 2, we take the point set
\[
	\{Q_1\oplus P| P\neq\mathcal{O},[15]P=\mathcal{O}\}\cup\{Q_2\oplus P| P\neq\mathcal{O},[15]P=\mathcal{O}\}.
\]
Denote the point set by $\{P_1,\ldots,P_n\}$ with $n=28$. Define two divisors $D=P_1+\cdots+P_n$ and $G=13\mathcal{O}+Q_1$. Then the code $C_L(D,G)$ is a $[28,14]$ iso-dual MDS code and
\[
	C_L(D,G)^{\perp}={\bf v}C_L(D,G)
\]
where ${\bf v}=(\frac{x+2}{h'_xy}(P_1),\ldots,\frac{x+2}{h'_xy}(P_n))$. Moreover, we have
\[
	{\rm dim}(C_L(D,G)\cap C_L(D,G)^{\perp})=0,
\]
which means that $C_L(D,G)$ is an LCD MDS code.
\end{example}

\begin{example}
Let $q=17^2$, and $\mathbb{F}_q^*=<\theta>$ where $\theta^2-\theta+3=0$. Consider the elliptic curve 
\[
\mathcal{E}:y^2=x^3+1
\]
defined over $\mathbb{F}_q$. We have $\#\mathcal{E}(\mathbb{F}_q)=324$ and $\mathcal{E}(\mathbb{F}_q)\simeq \mathcal{E}[2]\times \mathcal{E}[9]$. It follows that
\[
	Q_1=(-1,0),Q_2=(5\theta-2,0),Q_3=(-5\theta+3,0).
\]
From Construction 2, we take the point set
\[
	\{Q_2\oplus P|P\in \mathcal{E}[9]\backslash\{\mathcal{O}\}\}\cup\{Q_3\oplus P|P\in \mathcal{E}[9]\backslash\{\mathcal{O}\}\}.
\]
Denote the point set by $\{P_1,\ldots,P_n\}$ with $n=160$. Define two divisors $D=P_1+\cdots+P_n$ and $G=79*\mathcal{O}+Q_2$. Then the code $C_L(D,G)$ is a $[160,80]$ iso-dual MDS code and
\[
	C_L(D,G)^{\perp}={\bf v}C_L(D,G)
\]
where ${\bf v}=(\frac{x-5\theta+3}{h'_xy}(P_1),\ldots,\frac{x-5\theta+3}{h'_xy}(P_n))$. Moreover, we have
\[
	{\rm dim}(C_L(D,G)\cap C_L(D,G)^{\perp})=0,
\]
which means that $C_L(D,G)$ is an LCD MDS code.
\end{example}

\section{Application in construction of EAQECCs}
In this section, we will briefly review some definitions and results of EAQECCs. We mainly refer to \cite{ChenB} and \cite{Pereira}; see also in \cite{Brun} and \cite{Galindo}.
\begin{definition}
An entanglement-assisted quantum error correcting code (EAQECC) $\mathcal{Q}$ with parameters $[[n,k,d;c]]_q$ is an $q^k$-dimension subspace of $\mathbb{C}^{q^n}$ with minimum distance $d$ that consumes $c$ pre-shared copies of maximally entangled states. The rate of an EAQECC is given by $\frac{k}{n}$, the relative distance by $\frac{d}{n}$, and the entangled-assisted rate by $\frac{c}{n}$. Lastly, an EAQECC is said to consume or require maximal entanglement when $c=n-k$.
\end{definition}
The Singleton-type bound of an $[[n,k,d;c]]_q$ quantum codes was shown in \cite{Lai} (see also \cite{Grassl}). For $d\le\frac{n+2}{2}$, we have
\[
	2d\le n-k+c+2.
\]
An EAQECC is called an MDS EAQECC if the equality holds. The following results are useful in constructing MDS EAQECCs:
\begin{lemma}\cite[Lemma 16]{ChenB}\label{eaqeccs}
Let $C$ be an $[n,k,d]$ linear code over $\mathbb{F}_q$, and let $C^{\perp}$ be its dual with parameters $[n,n-k,d^{\perp}]$. Then there exists an EAQECC with parameters
\[
[[n,k-{\rm dim}(Hull(C)),d;n-k-{\rm dim}(Hull(C))]]_q.
\]
Moreover, if $C$ is MDS with $d\le\frac{n+2}{2}$, then the EAQECC is also MDS.
\end{lemma}
From Constructions 1 and 2, we can obtain the following results.
\begin{corollary}\label{mdseaqeccs}
Let $q$ be even. Then there exists an MDS EAQECC with parameters $[[n,c,d;c]]_q$ where $n\le\frac{q+\lfloor2\sqrt{q}\rfloor-1}{2}$ and $c\ge 1$.
\end{corollary}
\begin{proof}
Let $C$ be an $[n,k,d]$ iso-dual MDS code constructed by Construction 1. Note that $d=n-k+1$ and $n=2k$. Thus $d=k+1=\frac{n+2}{2}$. From Corollary \ref{hulleven} we have ${\rm dim}(Hull(C))\le k-1$. Using Lemma \ref{eaqeccs}, we can obtain an MDS EAQECC with parameters $[[n,c,d;c]]_q$ where $c\le 1$. With Remark IV.2, we have $n\le\frac{q+\lfloor2\sqrt{q}\rfloor-1}{2}$.
\end{proof}
\begin{corollary}
Let $q$ be odd. Then there exists an MDS EAQECC with parameters $[[n,c,d;c]]_q$ where $n\le\frac{q+\lfloor2\sqrt{q}\rfloor-3}{2}$ and $c\ge 1$.
\end{corollary}
\begin{proof}
The proof of MDS and $c\ge 1$ is similar to that of Corollary \ref{mdseaqeccs}. With Remark V.1, we have $n\le\frac{q+\lfloor2\sqrt{q}\rfloor-3}{2}$.
\end{proof}
\begin{example}
Using examples in Section IV and Section V, we have the following MDS EAQECCs:
\begin{itemize}
\item[1.] Let $q=16$, then there exists an MDS EAQECC with parameters $[[8,4,5;4]]_q$.
\item[2.] Let $q=32$, then there exists an MDS EAQECC with parameters $[[20,8,11;8]]_q$.
\item[3.] Let $q=64$, then there exists an MDS EAQECC with parameters $[[36,16,19;16]]_q$, and an MDS EAQECC with parameters $[[36,18,19;18]]_q$.
\item[4.] Let $q=256$, then there exists an MDS EAQECC with parameters $[[140,68,71;68]]_q$, and an MDS EAQECC with parameters $[[140,70,71;70]]_q$.
\item[5.] Let $q=25$, then there exists an MDS EAQECC with parameters $[[16,8,9;8]]_q$, and an MDS EAQECC with parameters $[[16,6,9;6]]_q$.
\item[6.] Let $q=49$, then there exists an MDS EAQECC with parameters $[[28,14,15;14]]_q$.
\item[7.] Let $q=289$, then there exists an MDS EAQECC with parameters $[[160,80,81;80]]_q$.
\end{itemize}
\end{example}
\section{Conclusion}
In this paper, we considered iso-dual MDS codes derived from elliptic curves. We first showed that the conjecture of maximal length MDS elliptic codes given by Han and Ren holds for iso-dual MDS elliptic codes. Then we provided two constructions of iso-dual MDS codes from elliptic curves. The length $n$ of codes constructed by these methods satisfies
\[
n\le\frac{q+\lfloor2\sqrt{q}\rfloor-1}{2}
\]
when $q$ is even and
\[
n\le\frac{q+\lfloor2\sqrt{q}\rfloor-3}{2}
\]
when $q$ is odd.

In the case where $q$ is even, an iso-dual code is equivalent to a self-dual code, and also equivalent to an LCD code. Therefore, the iso-dual code is useful for considering the hull-variation problem of linear code. However, constructing long self-dual MDS codes over finite fields with odd characteristics is still challenging. Similar to the even case, we can select points such that their values under $\frac{x-\beta_1}{h'_xy}$ are squares. But there is no explicit method available to find these points. Thus in $\mathbb{F}_q$ when $q$ is odd, the hull-variation problem remains an open problem, especially concerning increasing hull dimension through equivalent codes.

We also explored the hull dimension of the iso-dual codes presented in this paper. Subsequently, we applied them to the construction of EAQECCs, and obtained some MDS EAQECCs. It is difficult to claim that these EAQECCs have new parameters although the idea of iso-dual MDS elliptic codes in the previous literature has not been discussed. 

From the examples given in this paper, it appears that these iso-dual MDS codes have low hull dimensions. In the case of $q$ odd, these iso-dual MDS codes might even be LCD codes, and all codes equivalent to them also exhibit low hull dimensions. It is interesting to investigate the upper bound of iso-dual MDS codes, or simply iso-dual codes from algebraic curves. We will focus on this issue in our future work.

\section{Acknowledgment}
This work is supported by	Guangdong Major Project of Basic and Applied Basic Research (No. 2019B030302008), the National Natural Science Foundation of China (No. 12441107) and Guangdong Provincial Key Laboratory of Information Security Technology(No. 2023B1212060026).

\end{document}